\documentclass[a4paper]{article}
\pdfoutput=1

\usepackage[utf8]{inputenc}
\usepackage[nocompress]{cite}
\usepackage{fullpage}
\usepackage{hyperref}
\usepackage{amsmath, amssymb}
\usepackage{graphicx}
\usepackage[table]{xcolor} 
\usepackage{enumitem}
\usepackage{amsthm} 

\usepackage{multirow}

\newtheorem{theorem}{Theorem}

\newtheorem{lemma}{Lemma}

\newcommand{\Report}{\textsf{Report}}

\newcommand{\Update}{\textsf{Update}}

\newcommand{\problem}{\text{$w$-SSWSI}}
\newcommand{\dproblem}[1]{\text{($w,#1$)-SSWSI}}

\newcommand{\sn}{\ensuremath{s}}

\title{Sliding Window String Indexing in Streams}
\author{Philip Bille\footnote{Supported by Danish Research Council grant DFF-8021-002498.}
\\\texttt{phbi@dtu.dk} \and Johannes Fischer\\\texttt{johannes.fischer@cs.tu-dortmund.de} \and Inge Li G{\o}rtz\footnotemark[1]\\\texttt{inge@dtu.dk} \and Max Rish{\o}j Pedersen\footnotemark[1]\\\texttt{mhrpe@dtu.dk} \and Tord Joakim Stordalen\\\texttt{tjost@dtu.dk}}
\date{}

\begin{document}
\maketitle
\begin{abstract}
    Given a string $S$ over an alphabet $\Sigma$, the \emph{string indexing problem} is to preprocess $S$ to subsequently support efficient pattern matching queries, that is, given a pattern string $P$ report all the occurrences of $P$ in $S$. In this paper we study the \emph{streaming sliding window string indexing problem}. Here the string $S$ arrives as a stream, one character at a time, and the goal is to maintain an index of the last $w$ characters, called the \emph{window}, for a specified parameter $w$. At any point in time a pattern matching query for a pattern $P$ may arrive, also streamed one character at a time, and all occurrences of $P$ within the current window must be returned. The streaming sliding window string indexing problem naturally captures scenarios where we want to index the most recent data (i.e. the window) of a stream while supporting efficient pattern matching. 
        
    Our main result is a simple $O(w)$ space data structure that uses $O(\log w)$ time with high probability to process each character from both the input string $S$ and the pattern string $P$. Reporting each occurrence from $P$ uses additional constant time per reported occurrence. Compared to previous work in similar scenarios this result is the first to achieve an efficient worst-case time per character from the input stream with high probability. We also consider a delayed variant of the problem, where a query may be answered at any point within the next $\delta$ characters that arrive from either stream. We present an $O(w + \delta)$ space data structure for this problem that improves the above time bounds to $O(\log (w/\delta))$. In particular, for a  delay of $\delta = \epsilon w$ we obtain an $O(w)$ space data structure with constant time processing per character. The key idea to achieve our result is a novel and simple hierarchical structure of suffix trees of independent interest, inspired by the classic log-structured merge trees.  
\end{abstract}

\section{Introduction}
The \emph{string indexing problem} is to preprocess a string $S$ into a compact data structure that supports efficient subsequent pattern matching queries, that is, given a pattern string $P$, report all occurrences of $P$ within $S$. 
In this paper, we introduce a basic variant of string indexing called the \emph{streaming sliding window string indexing  (SSWSI) problem}. Here, the string $S$ arrives as a stream one character at a time, and the problem is to maintain an index of a \emph{window} of the last $w$ characters, for a specified parameter $w$. At any point in time a pattern matching query for a pattern $P$ may arrive, also streamed one character at a time, and we need to report the occurrences of $P$ within the current window. The goal is to compactly maintain the index while processing the characters arriving in either stream efficiently. We consider two variants of the problem: a \emph{timely} variant where each query must be answered immediately, and a \emph{delayed} variant where it may be answered at any point within the next $\delta$ characters arriving from either stream, for a specified parameter $\delta$. See Section~\ref{sec:setup_results} for precise definitions.

The  SSWSI problem naturally captures scenarios where we want to index the most recent data (i.e. the window) of a stream while supporting efficient pattern matching. For instance, monitoring a high-rate data stream system where we cannot feasibly index the entire stream but still want to support efficient queries. Depending on the specific system we may require immediate answers to queries, or we may be able to afford a delay that allows more efficient queries and updates.

The SSWSI problem has not been explicitly studied before in our precise formulation, but for the timely variant several closely related problem are well-studied. In particular, the \emph{sliding window suffix tree problem}~\cite{FG1989, Larsson1999, Senft2005,BJ2018, NAIP2003} is to maintain the \emph{suffix tree} of the current window (i.e., the compact trie of the suffixes of the window) as each character arrives. With appropriate augmentation the suffix tree can be used to process pattern matching queries efficiently, leading to a solution to the timely SSWSI problem. For constant-sized alphabets, the best of these solutions \cite{BJ2018} maintains the sliding window suffix tree in constant \emph{amortized} time per character while supporting efficient pattern matching queries. The other solutions achieve similar amortized time bounds. This amortization cannot be avoided since explicitly maintaining the suffix tree after the arrival of a new character may incur $\Omega(w)$ changes. 

Another closely related problem is the \emph{online string indexing problem}  \cite{AKLL2005, Kopelowitz2012, BI2013, Kosaraju1994, AN2008, KN2017, FG2015, AFGK+2014}. Here the goal is to process $S$ one character at a time (in either left-to-right or right-to-left order), while incrementally building an index on the string read so far. The best of these solutions achieve either constant time per character for constant-sized alphabets~\cite{KN2017} or $O(\log\log n + \log\log |\Sigma|)$ time for general alphabets~\cite{Kopelowitz2012} for building the index. These solutions all heavily rely on processing the string in right-to-left order to avoid the inherent linear time suffix tree updates due to appending, as mentioned above. Therefore they cannot be applied in our left-to-right streaming setting. Alternatively, we can instead apply these solutions on the reverse of the string $S$, but then each pattern must be processed in reverse order, which also cannot be done in our setting. Also, note that these solutions index the entire string read so far. It is not clear if they can be adapted to efficiently index a sliding window.

Our main result is an efficient solution to the SSWSI problem in both the timely and delayed variant. Let $w$ denote the size of the window. For the timely variant, we present a string index that uses $O(w)$ space and processes a character from the stream $S$ in $O(\log w)$ time. Each pattern matching query $P$ is also supported in $O(\log w)$ time per character with additional $O(\mathrm{occ})$ time incurred after receiving the last character of $P$, where $\mathrm{occ}$ is the number of occurrences of $P$ in the current window. The index is randomized and both time bounds hold with high probability. The results hold for any integer alphabet (not necessarily of polynomial size in $n$). Compared to the previous work, we improve the worst-case time bounds per character in the stream from $\Omega(w)$ to $O(\log w)$ with high probability. This is particularly important in the above mentioned applications, such as high-rate data stream systems. Our solution generalizes to the delayed variant of the problem. If we allow a delay of $\delta$ before answering each query we achieve $O(w + \delta)$ space while improving the above time bounds to $O(\log (w/\delta))$. In particular, if we allow a delay of $\delta = \epsilon w$ for any constant $\epsilon > 0$, we achieve linear space and optimal constant time (reporting the occurrences still takes $O(\mathrm{occ})$ time, and we do not count the reporting time towards the delay). All our results hold on a word RAM with logarithmic word size.

The key idea to achieve our result is a novel and simple hierarchical structure of suffix trees inspired by log-structured merge trees\cite{OCGO1996}. Instead of maintaining a single suffix tree on the window we maintain a collection of suffix trees of exponentially increasing sizes that cover the current window. We show how to efficiently maintain the structure as new characters from the stream arrive by incrementally ``merging'' suffix trees, while supporting efficient pattern matching queries within the window.

\subsection{Setup and Results}\label{sec:setup_results}
We formally define the problem as follows. Let $S$ be a stream over any integer alphabet $\Sigma$. For given integer parameters $w \geq 1$ and $\delta \geq 0$, the \emph{$\delta$-delayed streaming sliding window string indexing (\dproblem{\delta})} problem is to maintain a data structure that, after receiving the first $i$ characters of $S$, supports
    \begin{itemize}
        \item $\Report(P)$: report all the occurrences of $P$ in $S[i - w + 1 , i]$ before an additional $\delta$ characters have arrived, from either stream.
        \item $\Update()$: process the next character in the stream $S$. 
\end{itemize} 
In the $\Report(P)$ query the pattern string $P$ is also streamed. When $P$ is streamed it interrupts the stream $S$, arrives one character at a time, and all character of $P$ arrive before the streaming of $S$ resumes. Furthermore, we do not assume that we know the length of $P$ before the arrival of its last character. The delay is counted from after the last character of $P$ arrives. Each character from $S$ and from new patterns count towards the delay, while reporting occurrences does not (otherwise it would impossible to answer the query in time if there are more than $\delta$ occurrences).

We define the \emph{timely streaming sliding window string indexing (\problem{}) problem} to be $\dproblem{0}$, that is, queries must be answered immediately as the last character of the pattern arrives. 

We show the following general main result. 
\begin{theorem}
\label{thm:main_whp}
        Let $S$ be a stream and let $w \geq 1$ and $\delta \geq 0$ be integers. We can solve the \dproblem{\delta} problem on $S$ with an $O(w + \delta)$ space data structure that supports $\Update$ and $\Report$ in $O(\log \frac{w}{\delta + 1})$ time per character with high probability. Furthermore, $\Report$ uses additional worst-case constant time per reported occurrence.  
\end{theorem}
Here, with high probability means with probability at least $1 - \frac{1}{w^d}$ for any constant $d$. Theorem~\ref{thm:main_whp} provides a trade-off in the delay parameter $\delta$. In particular, plugging in $\delta = 0$ in Theorem~\ref{thm:main_whp} we obtain a solution to the timely SSWSI problem that uses $O(w)$ space and $O(\log w)$ time per character for both $\Update$ and $\Report$. Compared to the previous work~\cite{FG1989, Larsson1999, Senft2005, BJ2018, NAIP2003, ISTA2004, SD2008} this improves the worst-case bounds on the $\Update$ operation from $S$ from $\Omega(w)$ to $O(\log w)$ with high probability and also removes the restriction on the alphabet. At the other extreme, plugging in $\delta = \epsilon w$ for constant $\epsilon > 0$ in Theorem~\ref{thm:main_whp} we obtain a solution to the delayed SSWSI problem that uses $O(w)$ space and optimal constant time per character with high probability. All our results hold on a word RAM with logarithmic word size.

\subsection{Techniques}
We obtain our result for the timely variant, but without high probability guarantees, as follows. At all times we maintain at most $\log w$ suffix trees that do not overlap and together cover the window. The trees are organized by the \emph{log-structured merge technique}~\cite{OCGO1996}, where the rightmost tree is the smallest and their sizes increase exponentially towards the left. For each new character that arrives we append its suffix tree to the right side of our data structure. Whenever there are two trees of the same size next to each other we ``merge'' them by constructing a new suffix tree covering them both. Each character from $S$ is involved in at most $\log w$ merges and each merge takes expected linear time, so we spend expected amortized $O(\log w)$ time per character in $S$. We deamortize the updates by temporarily keeping both trees while merging them in the background. Note that for each adjacent pair of suffix trees we also store a suffix tree approximately covering them both, referred to as \emph{boundary trees} (see details below).

We find the occurrences of a pattern $P$ in the window by querying each of these trees, which takes $O(\log w)$ time per character in $P$. For adjacent pairs of trees larger than $|P|$ we find the occurrences of $P$ crossing from one into the other using the boundary trees. The remaining trees cover a suffix of the window of length $O(|P|)$, and we grow a suffix tree to answer queries in this suffix \emph{at query time}. Our data structure has some ``overhang'' on the left side of the window, and we use range maximum queries to report only the occurrences that start inside the window.

This solution is generalized to incorporate a delay of $\delta$ as follows. We store the $O(\log (w/\delta))$ largest trees from the timely solution and leave a suffix of size $\Theta(\delta)$ of the window uncovered by suffix trees. We answer queries as follows. If $|P| > \delta/4$ we say that $P$ is \emph{long}, and otherwise it is \emph{short}. For long patterns we do as in the timely case; the suffix tree we grow at query time now must also contain the uncovered suffix, but it still has size $O(|P|)$ since the uncovered part of the window has length $O(\delta) = O(|P|)$. We show how to do this in $O(\log (w/\delta))$ time per character in $P$. For short patterns we utilize that they are smaller than the delay to temporarily buffer the queries and later batch process them. We buffer up to $O(\delta\log(w/\delta))$ work and deamortize it over $\Theta(\delta)$ characters, obtaining the same bound as for long patterns. Updates run in the same bound since each character from $S$ is involved in at most $O(\log(w/\delta))$ merges before it leaves the window.

Finally, we improve the time bounds by proving that for any substring $S'$ of our window, we can construct the suffix tree over $S'$ in $O(|S'|)$ time with probability $1 - w^{-d }$ for any constant $d > 1$. We do so by reducing the alphabet $\Sigma' = \{c \in S'\}$ of $S'$ to rank-space $\{1,2,\ldots,|\Sigma'|\}$ from which the algorithm by Farach-Colton~et~al.~\cite{FFM2000} can construct the suffix tree in worst-case linear time. For large strings ($|S'| > w^{1/5}$) we pick a hash function $\Sigma \rightarrow [0,w^c]$ that with high probability is injective on $S'$, and then we use radix sort to reduce to rank-space in linear time. For small strings ($|S'| \leq w^{1/5}$) we pick a hash function $\Sigma \rightarrow [0,w/\log w]$ that is injective with (almost) high probability, and use this to manually construct a mapping into rank space in $O(S')$ time. This mapping algorithm uses additional $O(w\log w)$ space, but we construct at most $O(\log w)$ suffix trees at any time so the total space is linear. 

\subsection{Outline}
In Section~\ref{sec:preliminaries} we cover the preliminaries, including some useful facts about suffix trees. In Section~\ref{sec:real_time} we give a solution to the timely SSWSI problem that supports each operation in expected logarithmic time per character. In Section~\ref{sec:delay_ds} we show how to generalize this to incorporate delay, and in Section~\ref{sec:whp_ds} we show how to get good probability guarantees, proving Theorem~\ref{thm:main_whp}.

\section{Preliminaries}\label{sec:preliminaries}
 Given a string $X$ of length $n$ over an alphabet $\Sigma$, the $i$th character is denoted $X[i]$ and the substring starting at $X[i]$ and ending at $X[j]$ is denoted $X[i,j]$. The substrings of the form $X[i,n]$ are the \emph{suffixes} of $X$.

 A \emph{segment} of $X$ is an interval $[i,j] = \{i,i+1,\ldots,j\}$ for $1 \leq i \leq j \leq n$. We will sometimes refer to segments as strings, i.e., the segment $[i,j]$ refers to the string $X[i,j]$. The definition differs from ``substring'' by being specific about position; even if $X[1,2] = X[3,4]$ we have $[1,2] \neq [3,4]$. A \emph{segmentation} of $X$ is a decomposition of $X$ into disjoint segments that cover it. For instance, $x_1 = [1,i]$ and $x_2 = [i+1,n]$ is a segmentation of $X$ into two parts. The two segments $x_1$ and $x_2$ are \emph{adjacent} since $x_2$ starts immediately after $x_1$ ends, and for a pair of adjacent segments we define the \emph{boundary} $(x_1,x_2)$ to be the implicit position between $i$ and $i+1$. 

The \emph{suffix tree}~\cite{Weiner1973} $T$ over $X$ is the compact trie of all suffixes of $X\$$, where $\$ \not\in \Sigma$ is lexicographically smaller than any letter in the alphabet. Each leaf corresponds to a suffix of $X$, and the leaves are ordered from left to right in lexicographically increasing order. The suffix tree uses $O(n)$ space by implicitly representing the string associated with each edge using two indices into $X$. Farach-Colton et al.~\cite{FFM2000} show that the optimal construction time for $T$ is $\text{sort}(n, |\Sigma|)$, i.e., the time it takes to sort $n$ elements from the universe $\Sigma$.  For alphabets of the form $\Sigma = \{0,\ldots,n^c\}$ for constant $c \geq 1$ (that is, an integer alphabet of polynomial size in $n$) this implies that $T$ can be built in worst-case $O(n)$ time using radix sort. For larger alphabets we can reduce to the polynomial case in expected linear time using hashing, building $T$ in expected linear time (see Section~\ref{sec:whp_ds} for details). 

The \emph{suffix array} $L$ of $X$  is the array where $L[i]$ is the starting position of the $i$th lexicographically smallest suffix of $X$. Note that $L[i]$ corresponds to the $i$th leaf of $T$ in left-to-right order. Furthermore, let $v$ be an internal node in $T$ and let $s_v$ be the string spelled out by the root-to-$v$ path. The descendant leaves of $v$ exactly correspond to the suffixes of $X$ that start with $s_v$, and these leaves correspond to a consecutive range $[\alpha,\beta]_v$ in $L$.

We augment the suffix tree to support efficient pattern matching queries as follows. First, we use the well-known FKS perfect hashing scheme~\cite{FKS1984} to store the edges of the suffix tree, so we can for any node determine if there is an outgoing edge matching a character $a \in \Sigma$ in worst-case constant time. Note that this construction takes \emph{expected} linear time. Furthermore, we also build a \emph{range maximum query} data structure over $L$. This data structure supports range maximum queries, i.e., given a range $[\alpha,\beta]$ return the $j \in [\alpha,\beta]$ maximizing $L[j]$. It also supports range minimum queries, defined analogously. The data structure can be built in linear time and supports queries in constant time~\cite{GBT1984}. Finally, we preprocess the suffix tree in linear time such that each internal node $v$ stores the range $[\alpha,\beta]_v$ into $L$ corresponding to the occurrences of $s_v$. 

We can use this structure to efficiently find all the occurrences of $P$ in $O(|P| + \mathrm{occ})$ time, or the leftmost and rightmost occurrence of $P$ in $O(|P|)$ time. The \emph{locus} of a string $P$ is the minimum depth node $v$ such that $P$ is a prefix of $s_v$. Find the locus by walking downwards in the suffix tree, matching each character in $P$ in worst-case constant time using the dictionary. Once we have found $v$ we can report all the occurrences in $[\alpha,\beta]_v$ in $O(\mathrm{occ})$ time. Alternatively, we can find the rightmost occurrence of $P$ in constant time by doing a range maximum query on the range $[\alpha,\beta]_v$, which returns the $j \in [\alpha,\beta]_v$ maximizing the \emph{string position} $L[j]$. We can also find the leftmost occurrence by doing a range minimum query. 

Finally, note that it is possible to deamortize algorithms with \emph{expected} running time using the standard technique of distributing the work evenly. Specifically, if an algorithm runs in expected $\lambda n$ time we can do worst-case $\lambda$ work for $n-1$ steps; by linearity of expectation only expected $\lambda$ work remains for the last one.

\section{The Timely SSWSI Problem}\label{sec:real_time}
Here we present a solution for the timely variant that matches the bounds in Theorem~\ref{thm:main_whp} in expectation. Section~\ref{sec:whp_ds} shows how to get the bounds with high probability. Throughout this section we assume without loss of generality that $w$ is a power of two. Section~\ref{sec:real_time_amortized_updates} briefly mentions how to generalize to arbitrary $w$. 

The main idea is as follows.
We maintain a suffix of $S$ of length at least $w$.
This suffix is segmented into at most $\log w$ segments whose sizes are distinct powers of two, in increasing order from right to left.
The length of the suffix we store is at most $2^0 + \ldots  + 2^{\log w} = 2w - 1$. 
When a new character arrives, we append a new size-one segment to our data structure and merge equally-sized segments until they all have distinct sizes again.
We also discard the largest segment when it no longer intersects the window.
For each segment we store a suffix tree, and for every pair of adjacent segments we store a \emph{boundary tree} approximately covering them both (see below).
To support queries we query the suffix tree for each individual segment, and also each boundary tree.
%
For the segments larger than the pattern, the boundary trees are sufficient to find the occurrences crossing the respective boundary.
The remaining trees cover a suffix of $S$ that is $O(|P|)$ long, and we grow a suffix tree at query time to find the occurrences in this suffix. 

\subsection{Data Structure}\label{sec:real_time_ds}
At any point, the data structure contains a suffix $\sn$ of $S$ of length $w \leq |\sn| \leq 2w -1$ and a segmentation of $\sn$ into at most $\log w$ segments. Specifically, if $|\sn| = 2^{b_1} + \ldots + 2^{b_k}$ for integers $b_1 < \ldots < b_k$ then we have the segmentation $\sn_1,\ldots,\sn_k$ where $|\sn_i| = 2^{b_i}$, and $\sn$ is the concatenation of the strings $\sn_k,\sn_{k-1},\ldots,\sn_1$, in that order. The set $\{b_1,\ldots,b_k\}$ is unique and corresponds to the $1$-bits in the binary encoding of $|s|$. 
Three different configurations can be seen in Figure~\ref{fig:ds_log_example}.

For each segment $\sn_i$ we store the suffix tree $T_i$ over $\sn_i$, along with a range maximum query data structure over the suffix array of $\sn_i$. For each boundary $(\sn_{i+1},\sn_i)$ we store the \emph{boundary tree} $B_i$, which is the suffix tree over the substring centered at the boundary and extending $|\sn_i|$ characters in both directions. We augment $B_i$ with an additional data structure that we will use for reporting occurrences across the boundary. Let $\mathit{BL}_i$ be the suffix array corresponding to $B_i$. We define the \emph{modified suffix array} $\mathit{BL}'_i$ as 
\[
\mathit{BL}'_i[j] = \begin{cases}
    \mathit{BL}_i[i] & \text{if $\mathit{BL}[i]$ corresponds to a suffix starting in $s_{i+1}$}\\
    -\infty & \text{if $\mathit{BL}[i]$ corresponds to a suffix starting in in $s_{i}$}\\ 
\end{cases}
\]
We store a range maximum query data structure over $\mathit{BL}'_i$. Each of the data structures use $O(\sn_i)$ space, so the the whole data structures uses $O(\sn) = O(w)$ space. 

We note a few properties of the data structure. Let $S[n]$ be the most recent character to arrive and let $W_n = S[n-w+1,n]$ be the current window. Then $W_n$ is a suffix of $\sn$ since $|\sn| \geq w$. The largest, and leftmost, segment $\sn_k$ always has size $2^{\log w} = w$; it is not larger since $\log w$ bits are sufficient to represent $|s| \leq 2w - 1$, and it is always there since $|s| \geq w$ cannot be represented with $\log w - 1$ bits. For the same reason, $\sn_k$ always intersects at least partially with $W_n$, and each of $\sn_1,\ldots,\sn_{k-1}$ are fully contained in $W_n$. 
\begin{figure}
    \centering
    \includegraphics{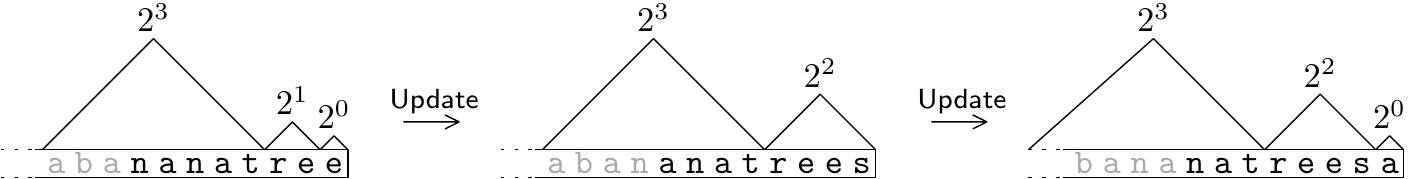}
    \caption{Example of updating the data structure with a window size of $w=8$. Here we illustrate the segments by the suffix trees built over them. Characters outside of the window are gray. As the character \texttt{s} arrives we construct a new suffix tree of size one, which is then immediately merged with the existing size-one suffix tree over \texttt{e} into a size-two suffix tree over \texttt{es}, which is then merged to into the final size-four suffix tree over \texttt{rees}. After receiving \texttt{a} we again have a size-one suffix tree. Note that after three more updates the suffix tree of size eight will no longer overlap the window and will be discarded.}
    \label{fig:ds_log_example} 
\end{figure}

\subsection{Queries}\label{sec:real_time_queries}
The idea is as follows, as exemplified in Figure~\ref{fig:query_summary}. Any occurrence of a pattern $P$ that is fully contained in a segment is found using the suffix tree over that segment. In the leftmost segment we must be careful to not report any occurrences that start before the left window boundary.
Similarly, any occurrence that only crosses a single boundary far enough away from the end of the window is found in the respective boundary tree.
The remaining occurrences are not contained in any of the trees in the data structure (either because they cross multiple boundaries or because they cross a single boundary $(s_{i+1},s_i)$ but start $>|s_i|$ characters to the left of the boundary). However, these occurrences are all located within a substring of size $O(m)$ ending at position $S[n]$, so we build, at query time, a suffix tree to find these occurrences. 

\begin{figure}
\centering
\includegraphics{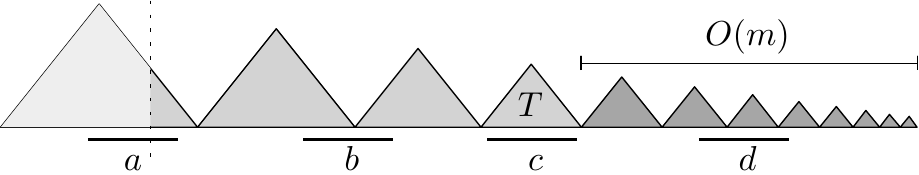}
\caption{Illustrates how we answer queries for a pattern $P$ of length $m$. The lines denoted $a$, $b$, $c$ and $d$ indicate occurrences of $P$. The segmentation is illustrated by the trees over the segments. The leftmost window boundary is marked with a vertical dashed line. Note that the leftmost segment intersects only partially with the window. The tree $T$ marks the smallest segment larger than $m$. The segments to the right of $T$ are all smaller than $P$, so they cover at most $m + m/2 + \ldots + 1 = O(m)$ characters. To answer the query for $P$ we match $P$ in the tree over each segment and in each boundary tree, and we also build a suffix tree over the segments smaller than $P$ at query time. We find $b$ because the respective boundary tree is sufficiently large. We find $c$ because it is fully contained in a segment. We find $d$ in the suffix tree that we build at query time. Note that $a$ is not contained in the window; we avoid reporting it by recursively using range maximum queries to find the \emph{rightmost} occurrence of $P$.}
\label{fig:query_summary}
\end{figure}

Let $P$ be the length-$m$ pattern being queried, $S[n]$ be the most recent character to arrive, and let $W_n$, the suffix $s$¸ the segmentation $\sn_1, \ldots, \sn_k$, and the indices $b_1 < \ldots < b_k$ be defined as above. 
As mentioned, any occurrence of $P$ in $W_n$ must either be fully contained within one of the segments, or it must cross the boundary between two adjacent segments.
We will show how to handle each of these cases separately. 

\paragraph{Fully Contained in a Segment}

Fix a specific segment $\sn_i$. As each character of $P$ arrives we match it in $T_i$. When the last character arrives we have a (possibly empty) range $[\alpha,\beta]$ into the suffix array of $\sn_i$ corresponding to the occurrences of $P$. If $\sn_i$ is not the leftmost segment then it fully contained in $W_n$ and we report all the occurrences. Otherwise, $\sn_i = \sn_k$ is the leftmost segment, which  might overlap only partially with $W_n$, and it may contain occurrences of $P$ that are not contained in the window. However, note that the intersection between $W_n$ and $\sn_k$ is a suffix of $\sn_k$. Therefore, if an occurrence of $P$ in $\sn_k$ \emph{starts} inside $W_n$ it also ends inside $W_n$. We find all such occurrences as follows. Let $L_k$ be the suffix array of $\sn_k$. As described in Section~\ref{sec:preliminaries} we find the index $j$ of the rightmost occurrence of $P$ by doing a range maximum query on the range $[\alpha,\beta]$ in $L_k$. If $L_k[j]$ is not inside $W_n$ then none of the occurrences are, and we are done. Otherwise we recurse on $[\alpha,j-1]$ and $[\beta,j-1]$. Matching $P$ in the trees of all the segments takes $O(\log w)$ overall time per character of $P$. Reporting each occurrence takes constant time since range maximum queries run in constant time.

\paragraph{Crossing a Boundary}
We show how to report the occurrences of $P$ that span a boundary. The main idea is as follows, as illustrated in Figure~\ref{fig:alternative_boundary_queries}. Let $\sn_i$ be the smallest segment where $|\sn_i| \geq m$. Consider any boundary $(\sn_{j+1},\sn_j)$ to the left of $\sn_i$, i.e., where $j \geq i$. Since both of these segments have size at least $|\sn_i| \geq m$, the boundary tree $B_j$ extends at least $m$ characters in both directions from the boundary. Therefore, all the occurrences of $P$ crossing the boundary are contained in $B_j$, and none of them can cross another boundary as well. Now consider the suffix $\mathcal{R}$ of $\sn$ containing the $m-1$ last characters of $\sn_i$ and extending to the end of $s$. This substring contains all the other boundary-crossing occurrences. Furthermore, all the occurrences in $\mathcal{R}$ cross at least one boundary since the longest consecutive part of a single segment in $\mathcal{R}$ is the $m-1$ characters in $\sn_i$. Note that the length of $\mathcal{R}$ is at most $m - 1 + |\sn_{i-1}| + |\sn_{i-1}|/2 + \ldots + 1 < m-1 + 2|\sn_{i-1}| < 3m$ since $|\sn_{i-1}| < m$. Thus, the number of boundary-crossing occurrences of $P$ equals the number of occurrences in $\mathcal{R}$ plus the number of occurrences crossing the boundaries $(\sn_{k},\sn_{k-1}), (\sn_{k-1},\sn_{k-2}),\ldots (\sn_{i+1},\sn_i)$.
    
The algorithm for finding the occurrences in the sufficiently large boundary trees is as follows. Fix a boundary $(s_{x+1},s_x)$. We match each character of $P$ in $B_x$ as it arrives. When the last character arrives we know if $|s_x| \geq m$, and also the range $[\alpha,\beta]$ corresponding to the occurrences of $P$ in the boundary tree. If $|\sn_x| \geq m$ (hence $x\ge i$) we report the occurrences as follows. As above we do a range maximum query to find the $j$ maximizing $\mathit{BL}_x'[j]$. If $\mathit{BL}_x'[j] = -\infty$ then all occurrences of $P$ start in $s_x$, and there are no occurrences crossing the boundary. Otherwise, $\mathit{BL}_x'[j]$ corresponds to the starting position of the rightmost occurrence of $P$ in $s_{i+1}$. Since all of $P$ has arrived and we now know $m$, we know that this occurrence crosses the boundary if and only if $\mathit{BL}'_x[j] \geq |\sn_x| - m + 2$ (recall that $B_x$ extends $|\sn_x|$ characters in both directions from the boundary). If it does not cross the boundary, then none of the other occurrences do either. Otherwise we report $\mathit{BL}'_x[j]$ and recurse on $[\alpha,j-1]$ and $[j+1,\beta]$ to find the remaining occurrences. Matching $P$ in all boundary trees takes $O(\log w)$ overall time per character, and reporting each occurrence with range maximum queries takes constant time.

We now show how to find the occurrences of $P$ in $\mathcal{R}$ with the same bounds. Assume that we know that $2^\ell \leq m < 2^{\ell+1}$ for some integer $\ell$. We build the suffix tree over the last $3\cdot2^{\ell+1}$ characters of $s$, deamortized over receiving the first $2^{\ell-1}$ characters of $P$. Over the next $2^{\ell-1}$ characters we match $P$ in the tree, at a rate of two characters per new character from $P$. Then, when the $2^\ell$th character arrives, we have caught up to the stream $P$, and we match the remaining $m-2^\ell$ characters as they arrive.  When the last character arrives we have matched $P$ in a tree of size at least $3m$, and we can start reporting occurrences. Note that we are overestimating the size of the tree, and it potentially includes some occurrences of $P$ that are contained in $\sn_i$. To avoid reporting these, we also build a range maximum query data structure over the suffix array such that we can use recursive range maximum queries. When deamortized, we construct the tree in expected constant time per character of $P$. Matching $P$ also takes constant time per character. We know that $m \leq w$, so we run this algorithm simultaneously for each of the $\log w$ different choices choices for $\ell$, using expected $O(\log w)$ time per character in $P$. Note that the trees use $O(w)$ space in total since the sum of the space it is a geometric sum where the largest term is $O(w)$.

\begin{figure}
    \centering
    \includegraphics[width=\linewidth]{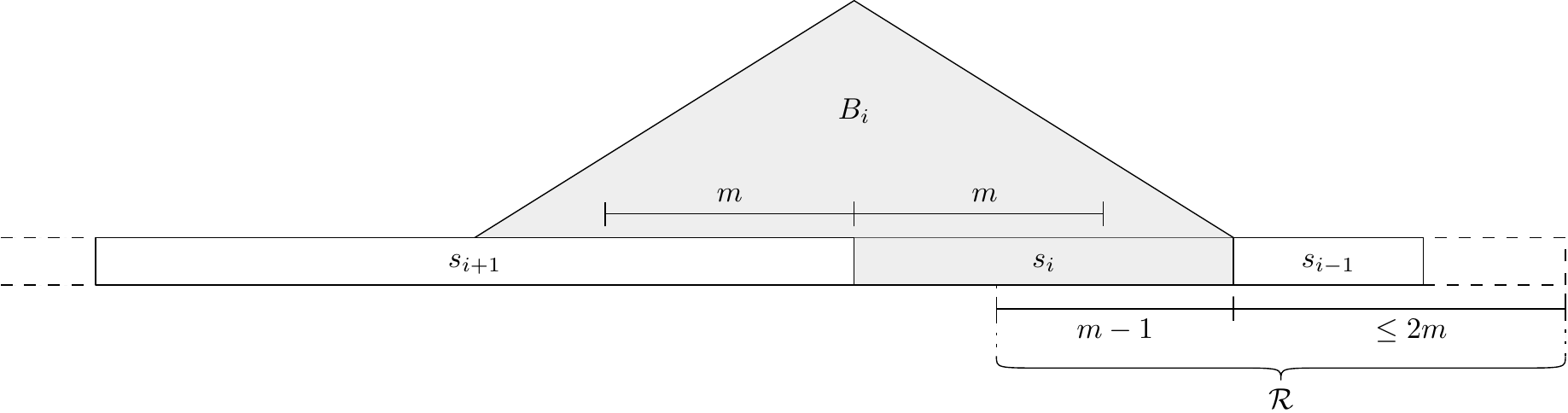}
    \caption{The segment $\sn_i$ is the smallest segment where $|\sn_i| \geq m$. For each boundary $(\sn_{j+1},\sn_j)$ where $j \geq i$, the tree $B_j$ is large enough to find all occurrences of $P$ across the boundary. All other occurrences of $P$ that cross a boundary must be in $\mathcal{R}$, the string covering the $m-1$ rightmost characters of $\sn_i$ and extending to the end of the window. The length of $\mathcal{R}$ is no more than $m-1 + |\sn_{i-1}| + |\sn_{i-1}| / 2 + \ldots + 1  < 3m $. }
    \label{fig:alternative_boundary_queries}
\end{figure}

\subsection{Amortized Updates}{\label{sec:real_time_amortized_updates}}
We show how to support updates in amortized $O(\log w)$ time. Let $S[n]$ be the last character to arrive and as in the description of the data structure let $b_1 < b_2 < \ldots < b_k$ be the positions of the $1$-indices in the binary encoding of $|\sn|$. When the new character $c = S[n+1]$ arrives, we update $\sn$ and the segmentation $\sn_1\ldots \sn_k$ to create the new suffix $\sn'$ with the new segmentation $\sn'_1,\ldots \sn'_{k'}$. See Figure~\ref{fig:ds_log_example} for an example.

 If $|s| < 2w-1$ then we set $s' = sc$. The segmentation of $s'$ corresponds to the unique binary encoding of $|s'| = |s| + 1$, so we update the segmentation analogously to a ``binary increment''. One way to do so is as follows. We create a new segment of size one over $c$. If there was not already a segment of size one, then we add the new segment and we are done. Otherwise we \emph{merge} (see below) the two size-one segments to create a segment of size two. The process cascades until we reach a size $2^b$ that does not exist in the segmentation of $\sn$ (i.e., the smallest index $b \not\in \{b_1,\ldots,b_k\}$). At this point we replace all of the segments $\sn_{b-1},\ldots,\sn_{1}$ with $\sn'_1$ covering the last $2^b$ characters of $\sn'$. The remaining segments for $\sn'$ are the same as the segments $\sn_{b+1},\ldots,\sn_k$. If $|\sn| = 2w - 1$ then there is a segment of each size $2^0, 2^1, \ldots, 2^{\log w}$. Since the segments have decreasing size from left to right, the {$\log w - 1$} rightmost segments cover the last $2^0 + \ldots + 2^{\log w - 1} = w-1$ characters of $\sn$. Thus, after $c$ arrives, the leftmost segment of size $2^{\log w} = w$ no longer intersects the window. We remove it by setting $\sn' = \sn[w+1,|\sn|]c$, and update the segmentation as above. 

 Let $a$, $b$ and $c$ be three adjacent segments, in that order. To \emph{merge} $b$ and $c$ we combine them into a new segment $d$ that spans them both, construct the suffix tree over $d$, and construct a range maximum query data structure on the suffix array of $d$. Furthermore, since $a$ and $d$ are now adjacent we also construct the boundary-spanning suffix tree for the boundary $(a,d)$ that extends $|d|$ characters in each direction. The construction of all of these data structures takes expected $O(|d|)$ time (see Section~\ref{sec:preliminaries}). Thus, it takes expected constant time per character every time it moves into a new, larger segment. Each character is contained in at most $\log w$ segments before it leaves the window, so the amortized update time is expected $O(\log w)$ per character. 

Note that all but the last merge are unnecessary to actually compute $\sn'_1$; in the amortized setting we can simply determine where the cascade will end and immediately construct the suffix tree over the corresponding segment. However, the cascading merges will come into play in the dearmotized variant.

Also note that if $w$ is not a power of two we can use a similar scheme where we allow either two simultaneous trees of size $2^{\lfloor \log w \rfloor}$, or one tree of size $2^{\lceil \log w \rceil}$. In both cases, there are some straightforward edge cases for when to remove the leftmost segment.

\subsection{Deamortized Updates}\label{sec:real_time_deamortized}
We now show how to deamortize the updates.
Unfortunately the previous construction cannot be directly deamortized since the suffix tree construction algorithm by Farach-Colton et al.~\cite{FFM2000} requires access to the whole string. Therefore, if a new character $c$ causes a cascade of merges resulting in a new segment of size $2^i$ we have to build the suffix tree over that segment when $c$ arrives.

Instead, we modify the structure slightly. When two segments of size $2^i$ become adjacent we temporarily keep both while deamortizing the cost of merging them over the \emph{next} $2^i$ characters of $S$, doing expected constant work per character. Note that queries are unaffected, with one exception for reporting occurrences across the boundaries; there might now be two adjacent segments $\sn_{i+1}$ and $\sn_i$ of the same size that are both the smallest segment at least as large as $|P|$. In this case the suffix $\mathcal{R}$ extends only $m-1$ characters into the rightmost segment $\sn_i$. The boundary tree for $(\sn_{i+1}, \sn_i)$ is large enough to report all occurrence crossing that boundary since both segments have size at least $|P|$. Furthermore, $\mathcal{R}$ potentially becomes twice as long, so we adjust the constants of the trees that we grow at query time.

To bound the time for updates we show that we are constructing at most $\log w$ suffix trees at any point, from which it follows that the update time is expected $O(\log w)$.
To do so we show the following lemma. 

\begin{lemma}\label{lem:deamortized-structure}
When the construction of a segment of size $2^i$ finishes there is exactly one segment of each size $2^{i-1}, \ldots, 2^0$.
\end{lemma}
\begin{proof}
    The proof is by induction on $i$. 
    For $i = 1$, when two size-one segments become adjacent we merge them when the next character $c$ from $S$ arrives. This results in a segment of size two, as well as a size-one segment containing $c$, proving the base case. 
    
    Inductively, consider the first time two segments of size $2^i$ become adjacent. 
    By the induction hypothesis, there is one segment of each size $2^0, 2^1, \ldots, 2^{i-1}$ to the right of these two segments. For another segment of size $2^i$ to be constructed, we must first receive one more character, which triggers a merge that eventually cascades through all $i-1$ of these segments. 
    For this to happen, $1 + (2^0 + 2^1 + \ldots + 2^{i-1}) = 2^i$ more characters from $S$ must arrive, where the $1$ is for the next character to arrive, and $2^j$ is the amount of characters the $j$th merge is deamortized over. However, at this point the merge of the two segments of size $2^i$ is complete, so we constructed two new segments, one of size $2^{i+1}$ and one of size $2^{i}$.
    By the induction hypothesis, there is also one segment of each size $2^0,\ldots, 2^{i-1}$, concluding the proof.  
\end{proof}

Lemma~\ref{lem:deamortized-structure} implies that there are never more than two segments of the same size adjacent to each other, and therefore at most one merging process for each segment size $2^0, 2^1, \ldots, 2^{\log w}$. To see this, consider the first time two segments $a$ and $b$ of size $2^i$ are adjacent. 
At this point, there are $2^0 + 2^1 + \ldots + 2^{i-1} = 2^i - 1$ characters to the right of $b$.
When the next segment $c$ of size $2^i$ arrives there are $2^i-1$ characters to the right of that, too. 
But then there are $|c| + 2^i-1 = 2^i + 2^i - 1$ characters to the right of $b$. Thus $2^i$ new characters must have arrived in the meanwhile, and the merging of $a$ and $b$ is done.

We obtain the following theorem.
\begin{theorem}{\label{thm:timely_result_expectation}}
Let $S$ be a stream and let $w \geq 1$ be an integer. We can solve the \problem{} problem with an $O(w)$ space data structure that supports $\Update$ and $\Report$ in expected $O(\log w)$ time per character. Furthermore, $\Report$ uses additional worst-case constant time per reported occurrence.      
\end{theorem}

\section{The Delayed SSWSI Problem}\label{sec:delay_ds}
In this section we show how to improve the result from Section~\ref{sec:real_time} if we are allowed a delay of $\delta$. The main idea is as follows. As before, we maintain suffix trees of exponentially increasing sizes, although only the $O(\log (w/\delta))$ largest of them. As a result there are fewer trees to query, but also an \emph{uncovered} suffix of size $\Theta(\delta)$ of the window for which we do not have any suffix trees. As in Section~\ref{sec:real_time} we denote the part of $S$ covered by suffix trees by $\sn$ and we denote the uncovered suffix by $t$. As above, $\sn$ is segmented into $\sn_1,\ldots,s_k$. 

We will first explain how to solve the problem when all patterns are \emph{long}, that is, $|P| > \delta/4$,  and then when all patterns are \emph{short}, that is, $|P| \leq \delta/4$. Finally we show how to combine these solutions. When all the patterns are long we can afford to construct, at query time, a suffix tree covering $t$. On the other hand, when all the patterns are short we can do both updates and queries in an offline fashion; we buffer queries and updates until we have approximately $\delta/2$ operations to do, at which point we can afford to construct a suffix tree over $t$ in a deamortized manner. See Figure~\ref{fig:delayed_queries} for an example.

Throughout this section we assume without loss of generality that $\delta$ is a power of two. Otherwise we instead use a more restrictive delay of $\delta' = 2^{\lfloor \log \delta \rfloor}$ and achieve the same asymptotic bounds.

\subsection{Long Patterns}
We first show how to support queries if all patterns have a length $m > \delta / 4$.
We modify the data structure from Section~\ref{sec:real_time} slightly. 
The smallest tree now has size $\delta/2$ as opposed to $1$, so there are $\Theta(\log w - \log(\delta/2)) = O(\log (w/\delta))$ segments and boundary trees. 
The  uncovered suffix $t$ has length at most $\delta$.

We answer queries the same way as in Section~\ref{sec:real_time_queries}, with only small modifications.
Let $P$ be a pattern of length $m > \delta / 4$.
As before, let $\sn_i$ be the smallest and rightmost segment with $|\sn_i| \geq m$.
We find any occurrence within a segment or crossing a single boundary by using the suffix trees over each segment and the boundary trees to the left of $\sn_i$, as before. 
The remaining occurrences we again find by growing suffix trees of exponentially increasing sizes from the right window boundary.
The only change is that we now grow the trees faster, as we must now also cover $t$, and we can afford to let the smallest tree have size $\delta$ since we have $m > \delta / 4$ characters in the pattern to deamortize the work over.
As above, let $\mathcal{R}$ be the string covering the $m-1$ last characters of $\sn_i$ and extending to the right window boundary, which now also includes $t$.
As $|t| < \delta$ the length of $\mathcal{R}$ is $|\mathcal{R}| < 3m + \delta < 7m$.
%
Assuming $2^{\ell} \leq m < 2^{\ell+1}$, we build the suffix tree of size $7 \cdot 2^{\ell+1}$ and match $P$ in it, amortized over the characters of $P$.
As we have $m > \delta / 4$ characters to deamortize the work over, we only do this for each choice of $\ell$ where $2^{\ell + 1} \geq \delta$, which results in $O(\log w - \log \delta) = O(\log (w / \delta))$ work per character in $P$.
As in Section~\ref{sec:real_time_queries} we use recursive range maximum queries to avoid double reporting any occurrences of $P$ that are also in $\sn$.

%
%
As there are also only $O(\log (w / \delta))$ segments and boundary trees we spend $O(\log (w / \delta))$ time per character in $P$. Note that we answer these queries without delay.

Updates are performed as follows. 
For each segment of $\delta/2$ characters that arrives we construct the suffix tree over it, deamortized over the next $\delta/2$ characters of $S$. 
We merge suffix trees as before, also deamortized over new characters of $S$.
The induction proof from Section~\ref{sec:real_time_deamortized} still works by modifying the base case; the merging of two trees of size $\delta/2$ takes $\delta/2$ characters, at which point another tree of size $\delta/2$ is constructed. The inductive step follows from the fact that $\delta$ is a power of two. Thus, we spend expected $O(\log(w/\delta))$ time per update.

\subsection{Short Patterns}
We now show how to support queries if all patterns have a length $m \leq \delta$.
We extend the data structure with a buffer of size $\delta$. 
This buffer will contain queries that we have not yet answered and characters for $S$ that we have not yet processed.
The total space is still $O(w + \delta) = O(w)$.
    
Whenever a character from $S$ arrives we append it to both $t$ and to the buffer. 
When a pattern arrives we append the full pattern to the buffer, and along with it we store the current position of the right window boundary.
Once the buffer has more than $\delta/2$ characters (patterns and text combined) we immediately allocate a new buffer of size $\delta$ and \emph{flush} the old buffer as follows. 
Note that at this point there are strictly less $\frac{3}{4}\delta$ characters in the buffer since each pattern is short. 

When we flush the buffer, we first answer all the buffered queries, and then we process all the buffered updates. 
We deamortize this work over the next $\delta/4$ characters that arrive from either stream. 

To answer the buffered queries we do as follows. 
Let $P_1,\ldots,P_\ell$ be the patterns in the buffer, let $m_i = |P_i|$, and let $M = \sum_{1\leq i \leq \ell} m_i$.
We have $M < \delta$. We start by building a suffix tree over $t$, along with a range minimum query data structure over the suffix array of $t$. This takes expected $O(\delta)$ time.
An occurrence of $P_i$ is either contained in $\sn$, or it crosses the boundary $(\sn,t)$, or it is contained in $t$. 
Since $P_i$ is smaller than each segment $\sn_j$ we can find all the occurrences within $\sn$ using the suffix trees over the segments and the boundary trees in $O(m_i\log(w/\delta))$ time. To find the occurrences crossing the boundary we build the KMP matching automaton~\cite{KMP1977} for $P_i$. In it we match the string that is centered at the boundary $(s,t)$ and extends $m_i-1$ characters in each direction. This takes $O(m_i)$ time. To find the the occurrences in $t$ we match $P_i$ in the suffix tree over $t$ in $O(m_i)$ time. In total, this takes $O(M\log(w/\delta)) = O(\delta\log(w/\delta))$ time for all the patterns, or expected $O(\log(w/\delta))$ time per character when deamortized. 
Note however, that after $P_i$ arrived more characters from $S$ could have arrived and been appended to $t$.
We must therefore take care not to report any occurrences of $P_i$ that extend past what \emph{was} the right window boundary when $P_i$ arrived. The KMP automaton finds the occurrences in left-to-right order, and in $t$ we avoid reporting too far right using recursive range minimum queries. 

Finally, we process each update in the buffer in the order they arrived, using the same procedure as for long patterns.
This takes $O(\log(w/\delta))$ time per update and $O(\delta\log(w/\delta))$ time in total.
Thus flushing the buffer takes expected $O(\log (w/\delta))$ time per character since we deamortize the expected $O(\delta\log(w/\delta))$ work over $\delta/4$ characters.
Since we allocate a new buffer immediately when we begin flushing, we will complete the flush before the next flush begins. 

\begin{figure}
    \centering
    \includegraphics[width=\textwidth]{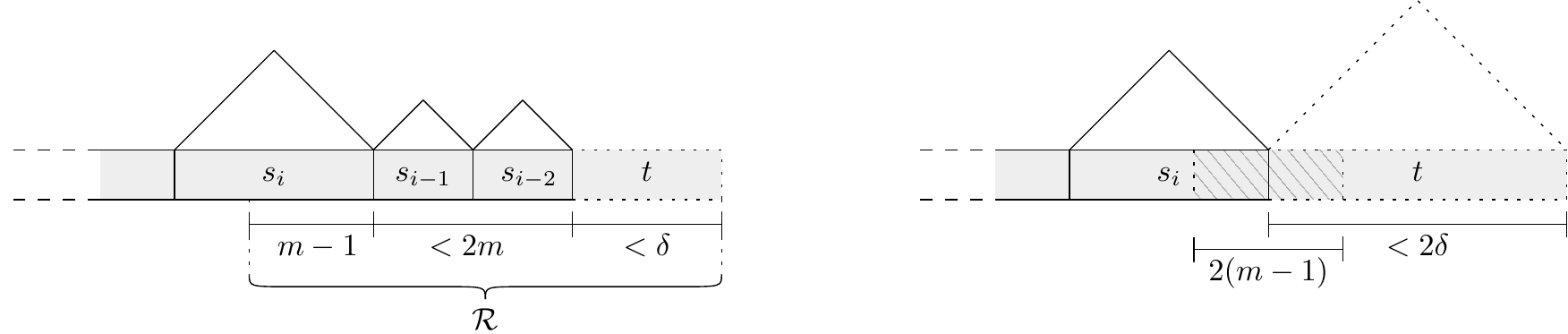}
    \caption{
    Left: Example of query with a long pattern. 
    Here $\sn_i$ is the smallest and rightmost segment with $|\sn_i| \geq m$.
    Note that the non-indexed suffix $t$ is less than $\delta < 4m$ characters long.
    Right: Example of a query with a short pattern.
    Note that for short patterns, $\sn_i$ is always the rightmost segment. 
    Any occurrence in $s$ cross at most a single boundary and is found using the constructed trees.
    Any occurrence in $t$ is found by the suffix tree over $t$ that we construct when we flush the buffer.
    Any occurrence that cross the boundary $(s, t)$ is found by the KMP automaton we build over the substring the extends $m-1$ characters in both directions from the boundary, which is hatched in the figure.
    } 
    \label{fig:delayed_queries}
\end{figure}

\subsection{Both Long and Short Patterns}
We now show how to combine the solutions short and long patterns, to obtain a solution that handles patterns of any length.
The data structure is the same as for small patterns above.
As above, we append each new character to the buffer.
However, whenever we start streaming a pattern we also proceed as if $P$ were long.
If $P$ turns out to fit in the buffer without triggering a flush (which might also happen if $P$ is long), we simply discard the work we did for the long-pattern case.
However, if adding $P$ to the buffer results in more than $\frac{3}{4}\delta$ characters being in the buffer, then $P$ must be long.
We immediately start flushing the buffer (ignoring the characters related to $P$) and also continue processing $P$ as a long pattern.
Note that since we are potentially streaming a long pattern while batch processing the updates in the buffer, the data structure might change while we are matching in it.
However, it only changes when a merge finishes, replacing a pair of suffix trees by a larger tree.
If this happens we keep the old trees in memory until we are done processing the pattern, at which point we discard them.

We obtain the following theorem.
\begin{theorem}{\label{thm:delayed_result_expectation}}
Let $S$ be a stream and let $w \geq 1$ and $\delta \geq 1$ be integers. We can solve the \dproblem{\delta} problem with an $O(w)$ space data structure that supports $\Update$ and $\Report$ in expected $O(\log (w/\delta))$ time per character. Furthermore, $\Report$ uses additional worst-case constant time per reported occurrence. 
\end{theorem}

\section{Obtaining High Probability}\label{sec:whp_ds}
In this section we show how to improve the time bounds to $O(\log (w/\delta))$ with probability $1 - w^{-d}$ for any constant $d \geq 1$. 

The expectation in the time bounds in Section~\ref{sec:delay_ds} comes from the construction of suffix trees (recall that we also build suffix trees at query time). Below, in Lemma~\ref{lem:suffix_tree_construction}, we prove that given a string $\mathcal{K}$ of length $k = O(w)$ we can construct the suffix tree over $\mathcal{K}$ in $O(k)$ time with probability $1 - 1/w^{1 + \epsilon}$, using additional $O(w/\log w)$ space. We use this algorithm to construct suffix trees during updates and queries, deamortizing them as before and doing $O(\log (w/\delta))$ work per character that arrives. When a new character arrives from $S$ or $P$, at most $O(\log(w/\delta)) = O(\log w)$ suffix tree constructions will finish. At this point, we finish constructing those trees that did not finish in time, that is, used more more time than what was allotted to them. By the union bound, the probability that any of them fail to finish in time (and thus incurring extra construction cost) is no more than $c\log w/w^{1+\epsilon}$ for some constant $c$ which is no more than $1/w$ for large $w$. Thus, for each character from $S$ or $P$ we spend $O(\log (w/\delta))$ time with high probability in $w$. We obtain the $1- 1/w^d$ probability bound by probability boosting, running $d = O(1)$ independent copies of the construction algorithm simultaneously. The algorithm from Lemma~\ref{lem:suffix_tree_construction} uses additional $O(w/\log w)$ space, but we are never constructing more than $O(\log w)$ suffix trees, so the space usage is $O(w)$ in total.

Furthermore, as mentioned in Section~\ref{sec:preliminaries}, we previously used an FKS dictionary~\cite{FKS1984} to store the edges to support reporting queries in worst-case constant time per character in the pattern.
The construction time of this dictionary is expected linear, so it can no longer be used.
Instead we use a dictionary by Dietzfelbinger and Meyer auf der Heide~\cite{DH1990}.
If there are $n$ elements in the dictionary it supports searches in worst-case constant time and any sequence of $\frac{1}{2}n$ updates takes constant time per update with probability $1 - 1/n^{d'}$ for any constant $d' \geq 1$.
We store all the edges of all the suffix trees in one such dictionary.
At all times, we keep $\Theta(w)$ dummy-elements in the dictionary to ensure that we get good probability bounds in terms of $w$, and we choose $d'$ large enough that any sequence of $O(w)$ operations (e.g., the construction of any one of our suffix trees) runs in $O(w)$ time with probability $1 - 1/w^{d + \epsilon}$.
     
\paragraph{Universal Hashing}
Before we prove Lemma~\ref{lem:suffix_tree_construction} we restate some basic facts about universal hashing, introduced by Carter and Wegman~\cite{CW1979}. Let $M,m > 0$ be integers, $\mathcal{H}$ be a set of functions $[0,M] \rightarrow [0,m]$, and $h \in \mathcal{H}$ be selected uniformly at random. Then  $\mathcal{H}$ is \emph{universal} if $P[h(x) = h(y) \mid x \neq y] \leq 1/m$.  Let $R \subseteq [0,M]$ and $|R| = r$. It follows from the union bound that $h$ has a \emph{collision} on $R$ with probability at most
    \begin{equation}\label{eq:prob_not_injective}
        P[h(x) = h(y) \text{ for some } x \neq y] \leq \sum_{x\neq y \in R} P[h(x) = h(y)] = \frac{r(r-1)}{2}\cdot \frac{1}{m} < \frac{r^2}{m}
    \end{equation}
In particular, if $m = r^c$ for constant $c \geq 1$ then $h$ is \emph{injective} (i.e., has no collisions) on $R$ with probability at least $1 - 1/r^{c-2}$. 
Carter and Wegman gave several classes of universal hash functions from which we can sample a function uniformly at random in constant time. 

\paragraph{Fast Suffix Tree Construction}   
We now prove Lemma~\ref{lem:suffix_tree_construction}, showing how to construct our suffix trees in linear time with high probability.

\begin{lemma}\label{lem:suffix_tree_construction}
Given a string $\mathcal{K}$ of length $k \leq 2w$ there is an algorithm that uses $O(k + w/\log w)$ space and constructs the suffix tree over $\mathcal{K}$ in $O(k)$ time with probability $1 - 1/w^{1+\epsilon}$ for some $\epsilon > 0$. 
\end{lemma}
\begin{proof}
Let $\sigma = \{\mathcal{K}[i] \mid i \in [1,k]\} \subseteq \Sigma$ be the alphabet of $\mathcal{K}$. We show how to, in $O(k)$ time, find a function $h: \Sigma \rightarrow [1,k^{O(1)}]$ such that $h$ is injective on $\sigma$ with probability at least $1 - 1/w^{1+\epsilon}$. If $h$ is injective on $\sigma$, we can construct the suffix tree over $\mathcal{K}'$ where $\mathcal{K}'[i] = h(\mathcal{K}[i])$ in time $O(\text{sort}(k,k^{O(1)})) = O(k)$ using radix sort. After the tree is constructed we can substitute for the original alphabet in linear time. Therefore, the construction algorithm finishes in $O(k)$ time with probability at least $1 - 1/w^{1+\epsilon}$ (otherwise we make no guarantee on the construction time and we can build the suffix tree in any way).
     
 For some $m$ to be determined later, let $f: \Sigma \rightarrow [1,m]$ be chosen uniformly at random from a class of universal hash functions. By Equation~\ref{eq:prob_not_injective}, the probability that $f$ has a collision on $\sigma$ is 
     \[
        P[f \text{ has collisions on } \sigma] < \frac{|\sigma|^2}{m} \leq \frac{k^2}{m}
     \] 
We divide into the cases of large trees ($k \geq w^{1/5}$) and small trees ($k < w^{1/5}$). If $k$ is large then $w^{1/5} \leq k \leq 2w$, and we set $m = w^4$ so the probability that $f$ has a collision is at most
     \[
     \frac{k^2}{m} \leq \frac{(2w)^2}{w^4} = \frac{4}{w^2} \leq \frac{1}{w^{1+\epsilon}}
     \] for some $\epsilon > 0$. 
We check whether $f$ is injective by sorting the set $\{(x,f(x)) \mid x \in \sigma\}$ with respect to the $f(\cdot)$-values and checking if two consecutive elements $(x,f(x))$ and $(y,f(y))$ has $x\neq y$ and $f(x) = f(y)$. This takes time $O(\text{sort}(k,w^4)) = O(k)$ using radix sort since $k \geq w^{1/5}$. If $f$ is injective we set $h = f$, concluding the proof of the large case. 
     
 If $k$ is small then we allocate an array $A$ of length $w/\log w$ in constant time. For simplicity we assume that $A$ is initialized such that $A[i] = 0$ for all $i$. This can be avoided using standard constant-time initialization schemes; assume each entry in $A$ contains an arbitrary value initially. We maintain two other arrays $B$ and $C$ such that if we have written a value to $A[i]$ at least once then $A[i]$ is a pointer to some $B[j]$, $B[j]$ is a pointer to $A[i]$, and $C[j]$ stores the value most recently written to $A[i]$. From this we can determine if $A[i]$ has been initialized (check if the pointers match) and if it has not, we can initialize it in constant time.
     
 Then we set $m = w/\log w$ such that the probability that $f$ has a collision is no more than 
     \[
         \frac{k^2}{m} < \frac{w^{2/5}}{w/\log w} = \frac{\log w}{w^{3/5}} = \frac{\log w}{w^{1/2}}\cdot\frac{1}{w^{1/10}} \leq \frac{1}{w^{1/10}}
     \] 
     for $w \geq 16$. We check if $f$ is injective on $\sigma$ by for each character $x$ in $\mathcal{K}$ setting $A[f(x)] = x$ and seeing if two distinct characters hash to the same index. If $f$ is injective we then arbitrarily assign the values $1,\ldots,|\sigma|$ to the now non-zero indices of $A$ and let $h(x) = A[f(x)]$ (at this point we know $\sigma$ since it is equal to the number of entries in $A$ that we modified). To boost the probability of success we run this algorithm up to eleven times with independent choices for $f$. The probability that all of them fail is at most $1/w^{11/10} \leq 1/w^{1+\epsilon}$ concluding the proof for the small case. 
\end{proof}

In conjunction with Theorems~\ref{thm:timely_result_expectation} and~\ref{thm:delayed_result_expectation}, this proves Theorem~\ref{thm:main_whp}.

\section{Conclusion and Future Work}
We have studied two variants of the streaming sliding window string indexing problem; the timely variant, where queries must be answered immediately, and the delayed variant where a query may be answered at any point within the next $\delta$ characters received, for a specified parameter $\delta$. For a sliding window of size $w$ we have given an $O(w)$ space data structure that supports updates in $O(\log w)$ time with high probability and queries in $O(\log w)$ time with high probability per character in the pattern; each occurrence is reported in additional constant time. For the delayed variant we improved these bounds to $O(\log (w/\delta))$, where each occurrence is still reported in constant time.

One open problem is whether these bounds can be improved. Another is to find efficient solutions when queries may be interleaved with new updates to the stream. That is, while you are streaming a pattern new characters of $S$ might arrive that move the current window.


\end{document}